\documentclass[12pt]{article}

\usepackage{graphicx}
\usepackage{multirow}%
\usepackage{amsmath,amssymb,amsfonts,amsthm}%
\usepackage{bm}%
\usepackage{mathrsfs}%
\usepackage{xcolor}%
\usepackage{textcomp}%
\usepackage{manyfoot}%
\usepackage{booktabs}%
\usepackage{algorithm}%
\usepackage{algorithmicx}%
\usepackage{algpseudocode}%
\usepackage{listings}%
\usepackage{float}

\usepackage[letterpaper,margin=1in]{geometry}

\newtheorem{theorem}{Theorem}%

\renewenvironment{abstract}
	{\quotation}
	{\endquotation}

\date{}

\makeatletter
\renewcommand{\fnum@figure}{\textbf{Figure \thefigure}}
\renewcommand{\fnum@table}{\textbf{Table \thetable}}
\makeatother

\usepackage{scicite}

\usepackage{url}

\def\scititle{
A discrete crack-tip theory for nonlinear lattice networks
}

\title{\large\bfseries \boldmath \scititle}

\author{
	Jiabin~Liu$^{1}$,
	Shaoting~Lin$^{1, 2\ast}$,
  Juntao~Huang$^{3\ast}$\and	
	\small$^{1}$Department of Mechanical Engineering, Michigan State University, East Lansing, MI, USA.\and
    \small$^{2}$Department of Mechanical Engineering, University of Wisconsin - Madison, Madison, WI, USA.\and
  \small$^{3}$Department of Mathematical Sciences, University of Delaware, Newark, DE, USA.\and
	\small$^\ast$Corresponding authors. Email: slin88@wisc.edu; huangjt@udel.edu.
}

\begin{document} 

\maketitle

\begin{abstract}
Crack-tip fields govern deformation localization and failure initiation. Classical continuum fracture mechanics describes these fields through theories such as the Hutchinson-Rice-Rosengren (HRR) field for nonlinear power-law solids. However, continuum descriptions break down near cracks in soft and architected materials, where load is transmitted through discrete chains, fibers, or struts. 
Here, we develop a discrete crack-tip theory for lattice networks with nonlinear chains. The theory has two central components. First, at large deformation, the strain of a representative chain in layer $i$ depends approximately linearly on the applied macroscopic strain, $\varepsilon_i\approx k_i(\lambda-1)$, defining a layer-dependent strain-amplification factor $k_i$. Second, along topology-selected chain directions, termed discrete HRR lines, the layer-to-layer ratios of $k_i$ follow a two-regime scaling law. 
Together, for a power-law chain force-strain relation with exponent $p$, the inner discrete regime predicts $\varepsilon_i\sim i^{-1/p}$ and $f_i\sim i^{-1}$, which differs from the classical continuum HRR prediction. The theory also explains why the intrinsic fracture energy approaches a size-independent limit as the network size increases. Photoelastic hydrogel experiments further validate our theory. These results reveal a two-regime crack-tip scaling law in nonlinear lattice networks and provide a framework for predicting chain deformation and intrinsic fracture energy.
\end{abstract}

\section{Introduction}\label{sec:intro}

Crack-tip fields characterize the local stress and deformation surrounding a crack tip and are central to fracture mechanics. In continuum fracture mechanics, analytical theories have established the stress intensity $K$-field of linear elastic fracture mechanics (LEFM) \cite{anderson2005fracture}, the Hutchinson-Rice-Rosengren (HRR) field for power-law-hardening elastoplastic solids \cite{hutchinson1968singular,rice1968plane}, and finite-deformation crack-tip fields for hyperelastic solids \cite{knowles1973asymptotic,knowles1974finite,stephenson1982equilibrium,geubelle1994finite,long2011finite,long2011effects,long2015crack,long2021fracture}.
These theories describe how stress and strain vary with distance and direction from a crack tip. In particular, for a continuum material governed by $\sigma\sim\varepsilon^p$, the HRR field predicts $\varepsilon\sim r^{-1/(p+1)}$ and $\sigma\sim r^{-p/(p+1)}$ with $r$ being the distance from the crack tip (Fig.~1a).

Continuum crack-tip fields, however, rely on a scale separation between the microscopic material scale and the scale over which the field varies.
This separation can break down near cracks in many soft and architected materials, where individual chains, fibers, or struts carry large deformation and initiate failure. As the crack tip is approached, the field may vary over distances comparable to the spacing between these load-bearing elements, making a continuum description insufficient.

Discrete network models therefore provide a natural framework for resolving the crack-tip mechanics of these materials. 
Studies of discrete fracture have shown that such bond-scale structure produces effects absent from continuum descriptions, including lattice trapping \cite{thomson1971lattice} and a pronounced dependence of fracture toughness and failure mechanisms on lattice topology and orientation \cite{fleck2007damage,nguyen2019role}. More recent studies have investigated crack fields and fracture resistance in elastoplastic lattices, soft cellular solids, polymer-network models, and architected materials undergoing large deformation \cite{tankasala2015cracktip,ma2016fracture,ghareeb2020adaptive,lei2021mesoscopic,bastek2025mode,yeerella2025fracture,huang2025topological}. Studies of polymer-like networks have further demonstrated that intrinsic fracture energy depends on chain nonlinearity and network heterogeneity \cite{deng2023nonlocal,hartquist2025scaling,hartquist2025fracture}. Nevertheless, a  framework connecting the deformation of the load-bearing chains near a crack tip to the macroscopic network deformation remains incomplete, particularly for networks composed of nonlinear chains undergoing large deformation.

This gap raises three questions. First, how does the deformation of each chain near the crack tip depend on the applied macroscopic stretch? Second, how do chain strain and force vary with layer number and lattice direction? Third, how does this spatial distribution control crack-tip deformation and intrinsic fracture energy as the network size increases?

Here, we address these questions using two-dimensional notched square and triangular lattice networks with  nonlinear chains and subjected to quasi-static pure-shear loading.
For these networks, our theory consists of two central components. First, we show that, at large deformation, the strain of a representative chain in layer $i$ depends approximately linearly on the applied macroscopic strain, $\varepsilon_i\approx k_i(\lambda-1)$, where $k_i$ is the layer-dependent slope, or strain-amplification factor. Second, along selected lattice directions, which we term discrete HRR lines, we develop a two-regime scaling law for the layer-to-layer slope ratios $k_i/k_{i+1}$, consisting of an inner discrete regime near the crack tip and an outer HRR-like regime far away from the crack tip. 

Together, the two components of our theory yields two main mechanical consequences. For power-law chains governed by $f\sim\varepsilon^p$, the inner regime predicts $\varepsilon_i\sim i^{-1/p}$ and $f_i\sim i^{-1}$, which differ from the classical continuum HRR predictions $\varepsilon\sim r^{-1/(p+1)}$ and $\sigma\sim r^{-p/(p+1)}$. The theory further predicts the size dependence of the crack-tip amplification, $k_1\sim N^{1/(p+1)}$, where $N$ is the number of layers in the network. This scaling explains why the intrinsic fracture energy $\Gamma_0$ converges to a size-independent limit as $N$ increases. Our photoelasticimitry experiments on hydrogel fabrics further validate the linear relation between local and macroscopic strain and the localization along topology-selected chain directions. These results reveal a topology-dependent two-regime scaling law in nonlinear lattice networks and establish a framework for predicting chain deformation and intrinsic fracture energy.

\section{Models}\label{sec:model}

We consider two-dimensional notched lattice networks with square and triangular topologies. Each network consists of nodes connected to their nearest neighbors by load-bearing chains. Each chain is modeled as a linear or nonlinear spring with undeformed length $l_0$ and force-stretch relation $f=f(\Lambda)$, where $\Lambda$ is the stretch ratio of that edge. Fracture is introduced by prescribing a critical stretch $\Lambda_f$ (equivalently a critical force $f_f=f(\Lambda_f)$), beyond which the spring breaks.

The nodal configuration of the network is denoted by $\{\boldsymbol{x}_i\}_{i=1}^{N_v}$, where $N_v$ is the total number of nodes. For each edge $e$ connecting node $\boldsymbol{x}_i$ and node $\boldsymbol{x}_j$, the stretch ratio is computed as $\Lambda_{e} = \|\boldsymbol{x}_i - \boldsymbol{x}_j\| / l_0$. The total elastic energy of the network is computed by summing the energy stored in all the springs,
\begin{equation}
  U = \sum_{e\in E} l_0 \int_1^{\Lambda_{e}} f(\Lambda) \, d\Lambda,
\end{equation}
where $E$ is the set of all edges in the network. At each loading step, the equilibrium configuration is obtained by minimizing $U$ with respect to the free nodal coordinates using the Fast Inertial Relaxation Engine (FIRE) algorithm \cite{bitzek2006structural}. Displacement boundary conditions are imposed on the top and bottom boundaries to perform quasi-static pure-shear loading. The network height changes from its undeformed value $h_0$ to $h$, and the macroscopic network stretch ratio is defined as $\lambda = h/h_0$.

The chain constitutive law is taken to be a monomial force-stretch relation:
\begin{equation}\label{eq:chain-constitutive-law}
  f = (\Lambda - 1)^p,
\end{equation}
where $p>0$ is a parameter, $\Lambda$ is the stretch ratio, and $f$ denotes the tensile force. The constitutive law \eqref{eq:chain-constitutive-law} is linear for $p=1$,  strain-stiffening for $p>1$, and strain-softening for $0<p<1$.
The intrinsic fracture energy $\Gamma_0$ is computed from the work done in an identical unnotched sample, integrated up to the critical height $h_c$ at the crack initiation (i.e., when a spring at the crack tip first reaches $\Lambda_f$ and breaks):
\begin{equation}
  \Gamma_0 = \int_{h_0}^{h_c} s \, dh,
\end{equation}
where $s$ is the nominal stress of the unnotched sample.

\section{Results}\label{sec:result}

\subsection{Linear scaling of stretch}

We first examine how the stretch of the near-tip chains evolves with the applied macroscopic network stretch. Let $\Lambda_i$ denote the stretch ratio of the representative chain in the $i$-th layer above the crack tip, with $i=1$ denoting the crack-tip chain, $i=2$ corresponding to the chain immediately above the crack tip, and so on (Fig.~\ref{fig:linear_scaling}a). For the monomial constitutive law (Fig.~\ref{fig:linear_scaling}b), the simulations show that, after an initial transient, the chain strain $\varepsilon_i=\Lambda_i-1$ is approximately linearly increasing with the macroscopic network strain $(\lambda-1)$ at large deformation:
\begin{equation}
    \Lambda_i-1 \approx k_i(\lambda-1),
    \label{eq:linear_scaling}
\end{equation}
where $k_i$ is a layer-dependent slope (Fig.~\ref{fig:linear_scaling}c and fig.~\ref{fig-supp:Figure_SI_1}). In addition, the crack-tip chain has the largest slope, indicating stretch concentration at the crack tip.

This linear scaling is attributed to an approximately frozen network geometry beyond a deformation threshold. At small deformation, the chains around the crack tip rotate and rearrange (Fig.~\ref{fig:linear_scaling}a). At large deformation, the angles between representative near-tip chains become nearly unchanged with further loading (Fig.~\ref{fig:linear_scaling}d). When the local chain angles are unchanged, the force balance at each node, together with the monomial constitutive law in Eq.~\eqref{eq:chain-constitutive-law}, implies that the stretch ratios of chains scale proportionally with $\lambda$. A formal statement of this argument is given in Theorem \ref{thm:fixed-angle-force-strain-ratios} in the appendix.

The linear slopes depend on both the network layer number $N$ and the nonlinearity index $p$. For fixed $p$, the slope increases with $N$, showing that larger networks generate stronger stretch concentration at the crack tip (Fig.~\ref{fig:linear_scaling}e). For fixed $N$, the slope decreases as $p$ increases (Fig.~\ref{fig:linear_scaling}f), suggesting weaker stretch concentration for strain-stiffening chains. This trend is consistent with the result reported in polymer-like networks with strain-stiffening chains \cite{deng2023nonlocal}.

The linear relation is not restricted to the monomial constitutive laws. We consider a polynomial chain constitutive law $f(\Lambda)=(\Lambda - 1)+(\Lambda - 1)^2$. In this case, the crack-tip stretch ratio still scales linearly with $\lambda$ at large deformation, and its slope is close to that of the quadratic monomial constitutive law $f(\Lambda)=(\Lambda - 1)^2$ (Fig.~\ref{fig:linear_scaling}g). This indicates that, at large deformation, the highest-order term in the chain constitutive law controls the near-tip scaling. 
Constitutive laws for polymer chains derived from statistical mechanics, including models involving the inverse Langevin function, can often be approximated by truncated power series expansions over finite stretch ranges \cite{arruda1993three,itskov2012taylor}. The same mechanism therefore applies whenever such a polynomial approximation remains accurate.

Finally, we test the triangular lattice networks (Fig.~\ref{fig:linear_scaling}h). Although the local crack-tip topology is more complex because no chain is exactly aligned with the vertical loading direction, representative chains still show approximately linear scaling at large deformation. The dependence of the slopes on $p$ remains the same: increasing $p$ weakens stretch concentration at the crack-tip chain (Fig.~\ref{fig:linear_scaling}i). Additional results for different values of $p$ and $N$ are shown in fig.~\ref{fig-supp:Figure_SI_2}.

\subsection{Discrete HRR-like scalings for stretch and force}

The linear relations in Eq.~\eqref{eq:linear_scaling} reduces the crack-tip stretch field to a sequence of slopes. We next quantify how these slopes vary with the distance from the crack tip. First, for fixed $p$, $k_i$ decreases with $i$, indicating that the stretch is concentrated around the crack tip (Fig.~\ref{fig:hrr_scaling}a and fig.~\ref{fig-supp:Figure_SI_3}). 
Then, we examine the chains near the crack tip. At large deformation, the two chains connected to the upper node of the crack-tip chain become nearly parallel and carry approximately equal tensile forces, denoted by $f_2$. Their resultant force balances the tensile force $f_1$ carried by the crack-tip chain (Fig.~\ref{fig:hrr_scaling}a), giving
\begin{equation}\label{eq:near-tip-force-balance}
  f_1 \approx 2 f_2.
\end{equation}
Combining this relation with the the constitutive law \eqref{eq:chain-constitutive-law} and the linear scaling \eqref{eq:linear_scaling} gives
\begin{equation}\label{eq:slope-ratio-crack-tip}
  k_1/k_2 \approx 2^{\frac{1}{p}}.
\end{equation}
Thus, the slope ratio between the first two layers is determined by the chain nonlinearity $p$ and independent of the network size $N$.

Motivated by the  relation \eqref{eq:slope-ratio-crack-tip} and the numerical observations for subsequent layers, we further propose the following scaling law for the inner discrete crack-tip regime:
\begin{equation}\label{eq:slope-ratio-inner-regime}
\frac{k_i}{k_{i+1}}
\approx
\left(1+\frac{1}{i}\right)^{1/p},
\qquad 1\leq i\leq n_0.
\end{equation}
At $i=1$, Eq.~\eqref{eq:slope-ratio-inner-regime} recovers the relation \eqref{eq:slope-ratio-crack-tip} derived from the local force balance at the crack tip. For $2\leq i\leq n_0$, it represents an empirical generalization supported by the numerical results. Here, $n_0$ denotes the approximate outer boundary of this regime and typically satisfies $3\le n_0\le 6$ in the present simulations. The exponent $1/p$ characterizes a discrete near-tip effect that is not captured by the continuum HRR theory.

Beyond the inner discrete crack-tip regime, the simulations reveal a different scaling:
\begin{equation}\label{eq:slope-ratio-outer-regime}
\frac{k_i}{k_{i+1}}
\approx
\left(1+\frac{1}{i}\right)^{\frac{1}{p+1}},
\qquad i>n_0.
\end{equation}
The exponent $1/(p+1)$ coincides with the strain exponent of the classical continuum HRR field for a constitutive relation $\sigma\sim\varepsilon^p$. We therefore refer to this range as the outer continuum HRR-like regime. Together, Eqs.~\eqref{eq:slope-ratio-inner-regime} and \eqref{eq:slope-ratio-outer-regime} define a two-regime scaling law along the discrete HRR line. The numerical slope ratios agree well with the proposed scaling across different network sizes $N$ and constitutive exponents $p$, as shown in Figs.~\ref{fig:hrr_scaling}b and \ref{fig:hrr_scaling}c, respectively.

Iterating Eq.~\eqref{eq:slope-ratio-inner-regime} gives the decay of the slopes within the inner discrete crack-tip regime:
\begin{equation}
{k_i}/{k_1}\approx i^{-1/p},
\qquad 1\leq i\leq n_0.
\end{equation}
Combining this relation with the constitutive law \eqref{eq:chain-constitutive-law} and the linear scaling \eqref{eq:linear_scaling}, we obtain the following layerwise scaling laws at a fixed macroscopic network stretch $\lambda$ in the regime of large deformation:
\begin{equation}
{\varepsilon_i}\sim i^{-1/p}, \qquad {f_i}\sim i^{-1}, \qquad 1\leq i\leq n_0.
\end{equation}
The strain profiles for different $p$ agree with this prediction over the first several layers (Figs.~\ref{fig:hrr_scaling}d and \ref{fig:hrr_scaling}e). Within this regime, increasing $p$ decreases the strain exponent $1/p$ and therefore weakens stretch localization. By contrast, the force exponent remains equal to one and is independent of $p$.

The different scaling law can be traced to continuum dimensional scaling \cite{chen2026thermodynamic}. In a homogeneous continuum, the energy release rate $G$ is energy per crack area, whereas the strain-energy density $W(r,\theta)$ is energy per unit volume. With $r$ as the only local length, dimensional consistency gives $W(r,\theta)\sim (G/r)\phi(\theta)$. For $\sigma\sim\varepsilon^p$, $W\sim\int_0^\varepsilon\sigma\,\mathrm{d}\varepsilon\sim\varepsilon^{p+1}$, giving $\varepsilon\sim r^{-1/(p+1)}$. A discrete network, however, introduces the microscopic length $l_0$. The continuum argument therefore requires a scale-separated region $l_0\ll r$, where $W$ averages over many chains. This condition fails within the first few lattice layers, where $r=O(l_0)$, and therefore cannot determine a local crack-tip field within this discrete core. There, the individual chain response and force balance give the new inner scaling $\varepsilon_i\sim i^{-1/p}$.

Moreover, whereas the continuum HRR theory defines a spatially continuous field with radial and angular dependence, the scaling in the discrete network is observed only along selected lattice directions. The predicted inner regime scaling is observed along the vertical discrete HRR line of the square lattice, whereas the horizontal crack-propagation line follows a different strain-decay profile (Fig.~\ref{fig:hrr_scaling}d and fig.~\ref{fig-supp:Figure_SI_7}).

Our two-regime scaling \eqref{eq:slope-ratio-inner-regime}-\eqref{eq:slope-ratio-outer-regime} also explains why the intrinsic fracture energy $\Gamma_0$ converges as the network size $N$ increases. For a square lattice with $N=2n$ vertical layers, the slopes satisfy the geometric height constraint $\frac{1}{2}k_1 + \sum_{i=2}^{n}k_i \approx n - \frac{1}{2}$.
Because $n_0$ remains of order unity as $N$ increases, the inner discrete regime contains only a finite number of chains. The size dependence is therefore controlled by the outer continuum HRR-like scaling:
\begin{equation}\label{eq:k0N_scaling}
    k_1 \sim N^{\frac{1}{p+1}}.
\end{equation}
This scaling is also numerically verified in Fig.~\ref{fig:hrr_scaling}f and fig.~\ref{fig-supp:Figure_SI_4}.

At crack initiation, the crack-tip chain reaches its critical stretch $\Lambda_f$. Eq.~\eqref{eq:linear_scaling} therefore gives
\begin{equation}\label{eq:lambdac}
    \lambda_c-1\approx\frac{\Lambda_f-1}{k_1} 
\end{equation}
For an unnotched square lattice subjected to the same pure-shear loading, $\Gamma_0$ scales as $N(\lambda_c-1)^{p+1}$. Assuming that $\Lambda_f$ is independent of $N$, Eqs.~\eqref{eq:k0N_scaling}-\eqref{eq:lambdac} give
\begin{equation}
\Gamma_0
\sim
N(\lambda_c-1)^{p+1}
\sim
N\left(N^{-1/(p+1)}\right)^{p+1}
=O(1).
\end{equation}
The increase in network size is therefore exactly offset, at the scaling level, by the decrease in the critical macroscopic deformation. This cancellation explains why $\Gamma_0$ converges to a size-independent limit. The detailed derivation is provided in Theorem~\ref{thm:Gamma0-convergence-scaling} in the appendix.

Triangular lattices exhibit the same two-regime scaling, although the relevant chain sequences are selected by different lattice topology. Along the direction $\theta=4\pi/3$, the numerical slope ratios agree with Eqs.~\eqref{eq:slope-ratio-inner-regime} and \eqref{eq:slope-ratio-outer-regime} for different $p$ (Fig.~\ref{fig:hrr_scaling}g). Comparisons along different directions $\theta=0$ and $\pi/3$ show that the predicted decay occurs only along selected direction (Fig.~\ref{fig:hrr_scaling}h). Along this discrete HRR line $\theta=4\pi/3$, the strains exhibit the same dependence on $p$ as those in square lattices (Fig.~\ref{fig:hrr_scaling}i). Our two-regime scaling is therefore not an artifact of vertical chain alignment in the square lattice, but a structure that persists across different lattice architectures.

\section{Experimental verification}\label{sec:experiment}

To experimentally validate our findings, we fabricate highly deformable, transparent polyacrylamide hydrogel networks with square and triangular lattice geometries. These networks not only sustain large deformations but also exhibit stretch-dependent color patterns under a circular polariscope, arising from stress-induced birefringence within their polymer networks\cite{liu2025fatigue}. This optical response makes these hydrogel networks suitable for directly visualizing the crack-tip field during loading. The fabrication procedure and the experimental setup are summarized in fig.~\ref{fig-supp:Exp_1}.

We first characterize the mechanical and photoelastic color responses of a single hydrogel fiber. The force-stretch relationship is nonlinear and it can be well described by a polynomial fit (Fig.~\ref{fig:experiment}a). Correspondingly, Fig.~\ref{fig:experiment}b shows the color changes of the hydrogel fiber with increasing stretch ratio, as captured under circularly polarized light. Fig.~\ref{fig:experiment}c shows a transparent undeformed square-lattice hydrogel network with 16 vertical layers which captured under white light. The deformed hydrogel networks at increasing stretch ratios of \(\lambda=1.5\), 2.0, and 2.5 are shown in Fig.~\ref{fig:experiment}d. Obviously, as the network stretch ratio increases, the chain at the crack tip exhibits the largest color change, while the deformation of the chains along the discrete HRR line decreases progressively with distance from the crack tip. We further compare the crack-tip fields of networks with different sizes, as shown in the top panels of Fig.~\ref{fig:experiment}d. At the same network stretch ratio, the crack-tip chain becomes more highly stretched as \(N\) increases, indicating stronger stretch concentration in larger networks.

We further quantify these observations by extracting the stretch ratio of each chain in the network from the captured images. Fig.~\ref{fig:experiment}e compares the measured stretch ratios of chains at layers \(i=1\), 2, and 3 along the discrete HRR line with simulations based on a polynomial fit to the force-stretch results of the hydrogel fiber. The experimental results agree closely with the simulations, and the chain stretch ratios increase approximately linearly with network deformation in the large-deformation regime. The results also confirm that the chain stretch decreases progressively with distance from the crack tip. Furthermore, we compare the crack-tip chain stretch ratios for networks of different sizes, for example the network layers is \(N=8\), 16, and 24, as shown in Fig.~\ref{fig:experiment}f. The quantitative results show that larger networks exhibit more pronounced stretch concentration at the crack tip, which is consistent with the observations in Fig.~\ref{fig:experiment}d. Additional polarized light images of square lattice networks with $N=8$ and $N=24$ are shown in fig.~\ref{fig-supp:Exp_Square}. Finally, in Fig.~\ref{fig:experiment}g, we fabricate triangular-lattice hydrogel networks and perform the same quantitative analysis. The extracted chain stretch ratios along the discrete HRR line of the triangular lattice also agree closely with the simulations.

\section{Discussion}\label{sec:discussion}

The main implication of this work is that lattice discreteness changes the structure of a crack-tip field. For a power-law chain response $f\sim\varepsilon^p$, the lattice layers near the crack tip exhibit the discrete scalings $\varepsilon_i\sim i^{-1/p}$ and $f_i\sim i^{-1}$, rather than the continuum HRR scalings $\varepsilon\sim r^{-1/(p+1)}$ and $\sigma\sim r^{-p/(p+1)}$. These discrete scaling laws are observed only along specific lattice directions, which we term the discrete HRR lines. These lines provide a discrete and topology-dependent description of the crack-tip field that is not captured by continuum theory.

The theory explains how the failure of a single chain at the crack tip is related to the fracture behavior of the entire network. In a larger network, deformation is more strongly concentrated at the crack tip, so less overall stretching is required to break the crack-tip chain. This reduction balances the additional elastic energy stored in the larger network, causing the intrinsic fracture energy to approach a size-independent limit. The theory therefore builds a connection between local chain failure and macroscopic fracture behavior.

The present theory is limited to ordered two-dimensional lattice networks under quasistatic pure-shear loading and focuses on crack initiation. Future work will extend the framework to more complex settings, including disordered and three-dimensional networks, more general chain interactions, and crack propagations.

\newpage
\begin{figure}[t]
    \centering
    \includegraphics[width=0.85\textwidth]{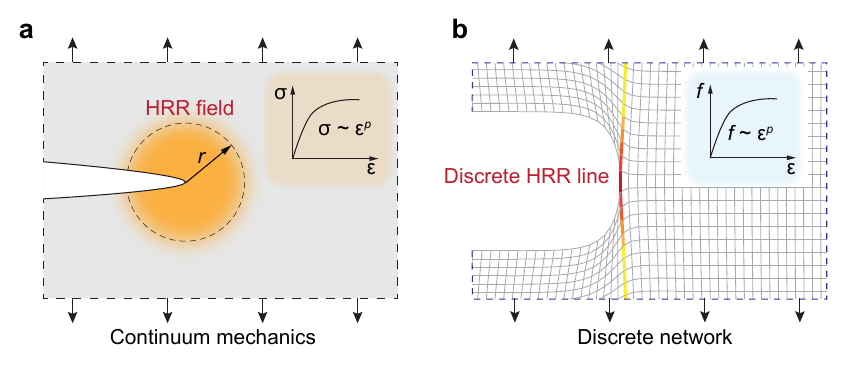}
    \caption{
    \textbf{From the continuum HRR field to a discrete HRR line.} 
    \textbf{(a)} A continuum material governed by the power-law constitutive relation $\sigma\sim\varepsilon^p$ develops a spatially continuous HRR field around the crack tip, where $r$ is the radial distance from the tip.
    \textbf{(b)} In a discrete lattice network with chain response $f\sim\varepsilon^p$, the chain strain and force exhibit layer-dependent scaling laws along certain lattice directions selected by the network topology. We refer to the chain family along such direction as a discrete HRR line. These scaling laws provide a discrete analogue of the continuum HRR field.}
    \label{fig:concept}
\end{figure}

\begin{figure}[t]
    \centering
    \includegraphics[width=\textwidth]{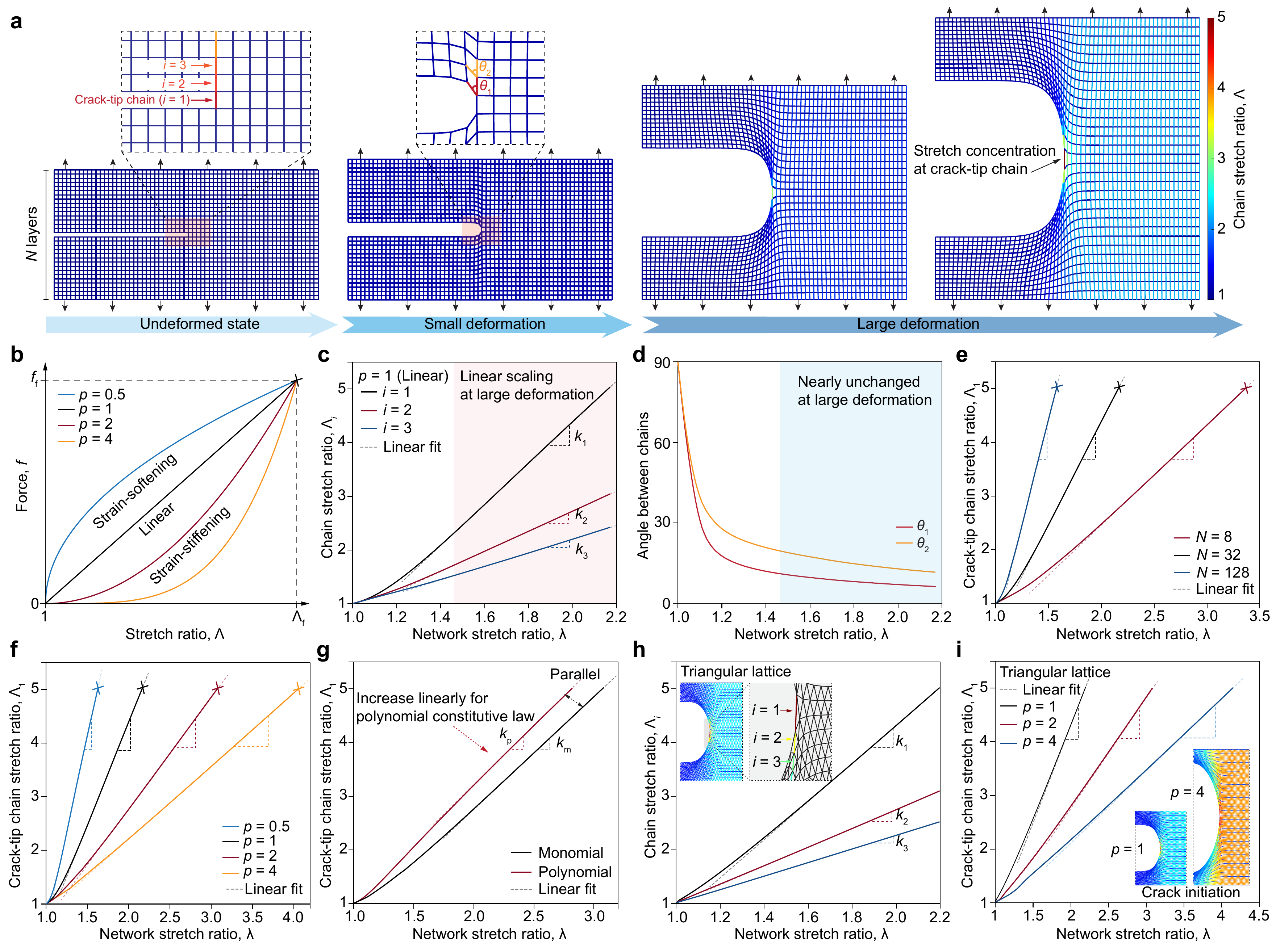}
    \caption{\textbf{Linear scaling of chain stretch ratio.} 
    \textbf{(a)} Demonstration of the pure-shear loading on a notched square lattice network from the undeformed state to the deformed state till the crack initiation. 
    \textbf{(b)} Monomial chain constitutive law, i.e., force-stretch relation $f=(\Lambda-1)^p$ for strain-softening ($p<1$), linear ($p=1$), and strain-stiffening ($p>1$) chains.
    \textbf{(c)} 
    At large deformation, the chain strains become approximately linear functions of the macroscopic network strain, $\Lambda_i-1\approx k_i(\lambda-1)$, where $k_i$ is the layer-dependent slope.
    \textbf{(d)} The local angles between representative chains become nearly unchanged at large deformation, providing a geometric explanation for the linear stretch scaling.
    \textbf{(e)} The linear slope increases with the network layer number $N$. 
    \textbf{(f)} The linear slope decreases with the nonlinearity exponent $p$. 
    \textbf{(g)} A polynomial chain constitutive law $f=(\Lambda-1)+(\Lambda-1)^2$ shows the same  linear scaling, with the leading nonlinear term controlling the slope $k_{\textrm{p}} \approx k_{\textrm{m}}$. \textbf{(h, i)} Similar linear scaling is observed in triangular lattices.}
    \label{fig:linear_scaling}
\end{figure}

\begin{figure}[t]
    \centering
    \includegraphics[width=\textwidth]{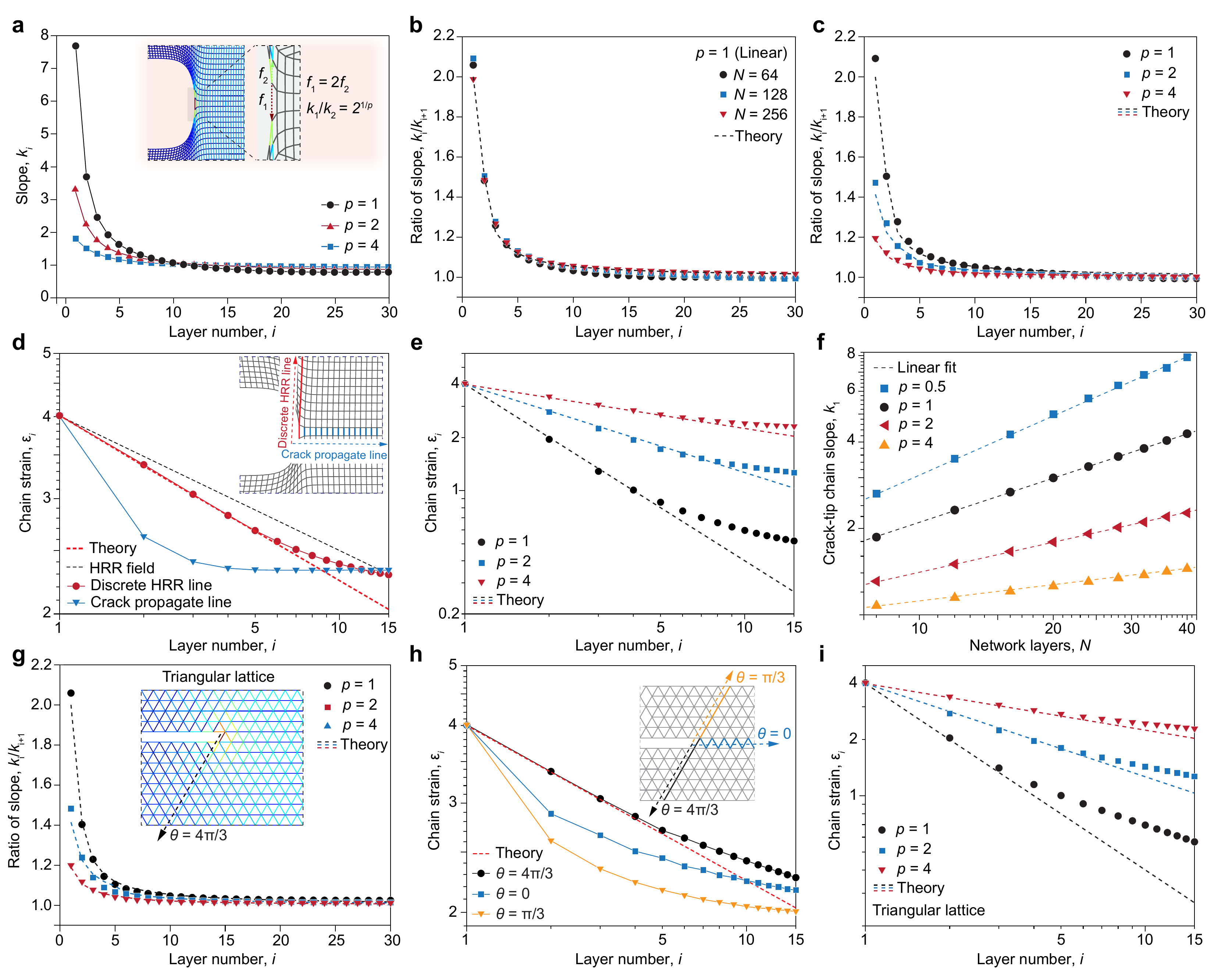}
    \caption{\textbf{Two-regime discrete HRR-like scaling in square and triangular lattices.}
    \textbf{(a)} The slopes $k_i$ decrease with layer number $i$, showing stretch concentration near the crack tip. The local force balance near the crack tip gives $f_1\approx2f_2$ and hence $k_1/k_2\approx2^{1/p}$.
    \textbf{(b, c)} The slope ratios $k_i/k_{i+1}$ for linear chains and different network sizes. The theory curves show the inner discrete scaling $(1+1/i)^{1/p}$ and outer HRR-like scaling $(1+1/i)^{1/(p+1)}$ in Eqs.~\eqref{eq:slope-ratio-inner-regime}-\eqref{eq:slope-ratio-outer-regime}.
    \textbf{(d)} Chain-strain decay along the vertical discrete HRR line and horizontal crack-propagation line.
    \textbf{(e)} Strain profiles for different $p$; dashed lines show the inner scaling $\varepsilon_i\sim i^{-1/p}$.
    \textbf{(f)} Network size dependence of the crack-tip slope, $k_1\sim N^{1/(p+1)}$.
    \textbf{(g-i)} Corresponding slope-ratio, directional, and strain-scaling results for triangular lattices, with $\theta=4\pi/3$ identifying the topology-selected discrete HRR line.
    }
    \label{fig:hrr_scaling}
\end{figure}

\begin{figure}[t]
    \centering
    \includegraphics[width=\textwidth]{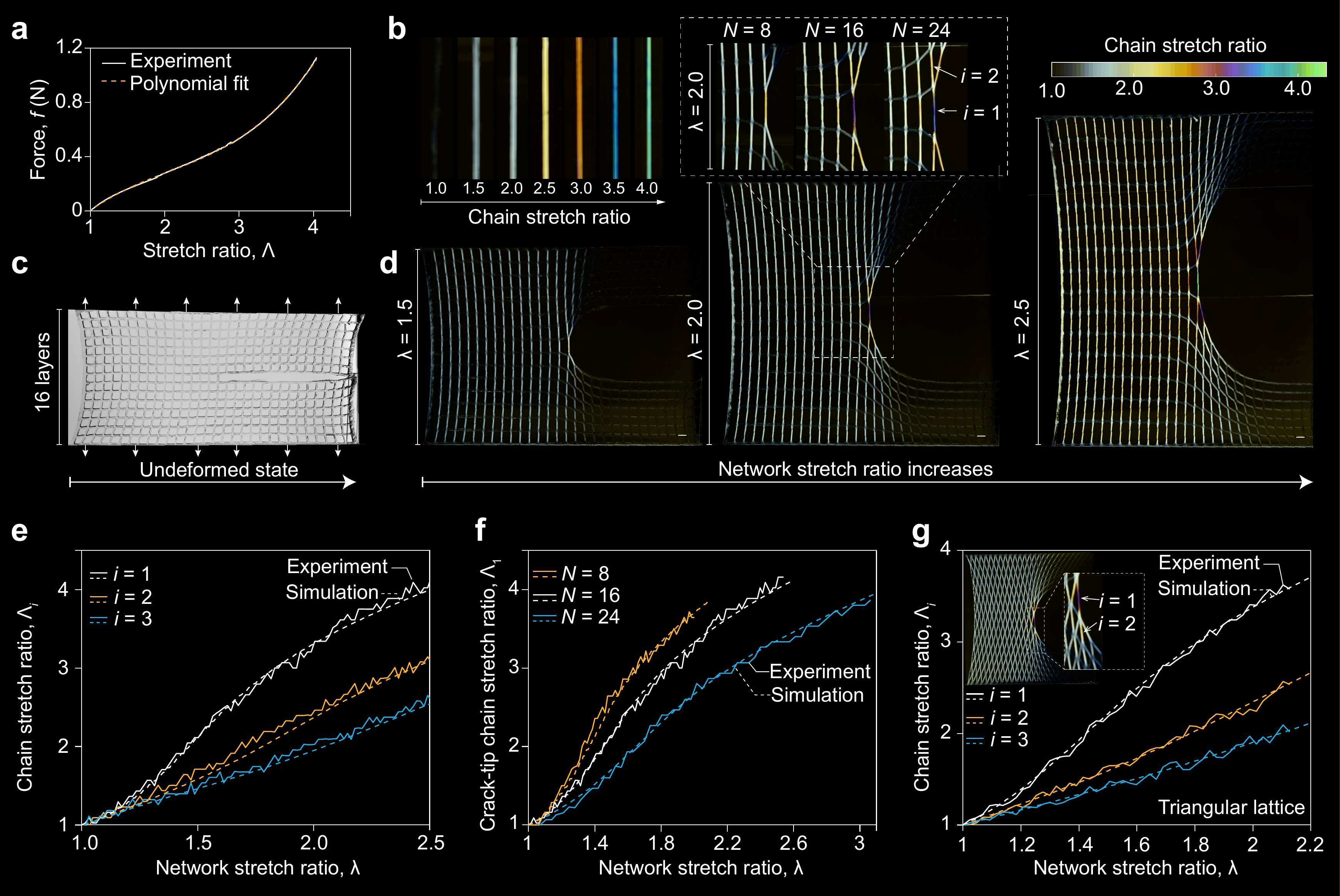}
    \caption{
    \textbf{Experimental validation using hydrogel networks.} \textbf{(a)} Force-stretch curve of a single hydrogel fiber. \textbf{(b)} Images of a single hydrogel fiber under polarized light at increasing stretch ratios. \textbf{(c)} Image of an undeformed square-lattice hydrogel network captured under white light. Scale bar: 10 mm. \textbf{(d)} Images of deformed networks captured under polarized light at stretch ratios $\lambda=1.5$, $2.0$, and $2.5$. The upper insets show the crack-tip fibers at the same stretch ratio, $\lambda=2.0$, for networks with $N=8$, 16, and 24 layers. \textbf{(e)} Chain stretch ratio versus network stretch ratio for chains along the discrete HRR line in a square-lattice hydrogel network. \textbf{(f)} Crack-tip chain stretch ratio versus network stretch ratio for square-lattice hydrogel networks with $N=8$, 16, and 24 layers. \textbf{(g)} Chain stretch ratio versus network stretch ratio for chains along the discrete HRR line in a triangular-lattice hydrogel network. Experimental measurements are compared with simulations in (e-g).
    }
    \label{fig:experiment}
\end{figure}

\clearpage %

\bibliography{ref} %
\bibliographystyle{sciencemag}

\section*{Acknowledgments}

\paragraph*{Funding:}
J.H. was partially funded by the National Science Foundation (DMS-2618114) and the startup fund of the College of Arts and Sciences at University of Delaware. J.L. and S.L. was partially funded by the National Science Foundation (CMMI-2338747), the National Science Foundation (CMMI-2423067), and the startup fund of the College of Engineering at  Michigan State University. Support for this research was also provided by the University of Wisconsin-Madison, Office of the Vice Chancellor for Research with funding from the Wisconsin Alumni Research Foundation.
\paragraph*{Author contributions:}
J.L., S.L. and J.H. conceived the idea. J.L. and J.H. conducted the numerical simulations. J.L. conducted the experiments and prepared the figures. J.L., S.L. and J.H. interpreted the results and wrote the manuscript. S.L. and J.H. acquired the funding and supervised the project.
\paragraph*{Competing interests:}
There are no competing interests to declare.
\paragraph*{Data and materials availability:}
The code used for numerical simulations will be made publicly accessible upon publication of the paper.

\subsection*{Supplementary materials}
Materials and Methods\\
Theory\\
Figures S1 to S7\\

\newpage

\renewcommand{\thefigure}{S\arabic{figure}}
\renewcommand{\thetable}{S\arabic{table}}
\renewcommand{\theequation}{S\arabic{equation}}
\renewcommand{\thepage}{S\arabic{page}}
\renewcommand{\thetheorem}{S\arabic{theorem}}
\renewcommand{\thedefinition}{S\arabic{definition}}
\renewcommand{\thealgorithm}{S\arabic{algorithm}}
\renewcommand{\theremark}{S\arabic{remark}}
\renewcommand{\theassumption}{S\arabic{assumption}}
\renewcommand{\thelemma}{S\arabic{lemma}}
\renewcommand{\thecorollary}{S\arabic{corollary}}
\setcounter{figure}{0}
\setcounter{table}{0}
\setcounter{equation}{0}
\setcounter{theorem}{0}
\setcounter{definition}{0}
\setcounter{page}{1} %
\renewcommand{\thesection}{S\arabic{section}}
\renewcommand{\thesubsection}{S\arabic{section}.\arabic{subsection}}

\begin{center}
\section*{Supplementary Materials for\\ \scititle}

Jiabin~Liu,
Shaoting~Lin$^{\ast}$,
Juntao~Huang$^{\ast}$
\\ %
\small$^\ast$Corresponding authors. Email: slin88@wisc.edu; huangjt@udel.edu\\
\end{center}

\subsubsection*{This PDF file includes:} 
Materials and Methods\\
Supplementary Text\\
Figures S1 to S7\\
References

\newpage

\section*{Materials and Methods}
\textbf{Materials}

\noindent Acrylamide (AAm, Sigma-Aldrich A8887), N,N$'$-Methylenebisacrylamide (MBAA, Sigma-Aldrich 146072), ammonium persulfate (APS, Sigma-Aldrich A3678), and N,N,N$'$,N$'$-tetra\-methyl\-ethylene\-diamine(TEMED, Sigma-Aldrich T9281) used in this work were purchased from Sigma-Aldrich and used without modification. Transparent acrylic sheets (8560K191 and 8560K171) used to fabricate hydrogel molds and conduct mechanical tests were purchased from McMaster-Carr. The linear polarizing film (XP44-40) and $\lambda/4$ retarder film (WP140HE) used in the homemade photoelastic setup were purchased from Edmund Optics. For the homemade photoelastic testing setup, a white LED panel light with a color temperature of 4,000 K was purchased from The Home Depot, and a Nikon D800 camera was used to capture images. 

\vspace{1em} 
\noindent\textbf{Synthesis of Hydrogel Network}

\noindent Before synthesizing the hydrogel networks, we printed the network molds by PLA using a 3D printer (Bambu Lab A1). The PAAm pre-gel solution was prepared by dissolving acrylamide (AAm) monomer in deionized water at a monomer-to-water ratio of 3:1. A 0.1 M ammonium persulfate (APS) solution was used as the thermal initiator, 0.23 wt\% N,N$'$-methylenebisacrylamide (MBAA) was used as the crosslinker, and N,N,N$'$,N$'$-tetramethyleth
ylenediamine (TEMED) was used as the polymerization accelerator. For each 10 g batch of PAAm pre-gel solution, 200 $\mu$L of MBAA solution, 150 $\mu$L of APS solution, and 10 $\mu$L of TEMED were added. The mixture was vortexed for 2 min and then degassed in a vacuum chamber for an additional 2 min. The precursor solution was carefully poured into the 3D-printed molds while avoiding the formation of air bubbles, and the molds were covered with an acrylic sheet. The hydrogel networks were cured for 3 h under ambient conditions. After curing, the networks were immersed in deionized water until swelling equilibrium was reached.

\vspace{1em} 
\noindent\textbf{The Circular Polariscope}

\noindent The circular polariscope is a photoelastic setup used to directly visualize deformed birefringent materials.  The setup consists of an LED light panel (4000\,K), two linear polarizers, two quarter-wave plates, a universal testing machine, and a camera. The first polarizer and quarter-wave plate (collectively referred to as the polarizer) are placed before the sample, while the second set (the analyzer) is placed after. The two linear polarizers are installed orthogonally to each other to create the darkest possible observed light field. The first quarter-wave plate is installed at a 45-degree angle, and the second quarter-wave plate is orthogonal to the first to darken the observed light field again. 

\vspace{1em} 
\noindent\textbf{Mechanical Tests of Hydrogel Network}

\noindent We first tested the force–stretch relationship of individual hydrogel fibers.The hydrogel fiber was cut to a length of 20 mm and subjected to uniaxial tensile loading. The tests were conducted under a custom-built circular polariscope, and the photoelastic images were captured by the camera at one-second intervals. Before mechanical testing, the top and bottom of each hydrogel network were glued to acrylic sheets using superglue. Before testing, the samples were briefly immersed in water to ensure that they remained fully swollen. Then each sample was then stretched at a constant rate of 1 mm/s by the universal testing machine. Photoelastic images were captured through the circular polariscope at 1 s intervals, while the force was simultaneously recorded by the mechanical testing machine.

\section*{Theory}

\begin{theorem}[Fixed angle force-strain ratios]\label{thm:fixed-angle-force-strain-ratios}
Consider a node $v$ connected to $m\ge 2$ springs. For each  spring $i=1,\dots,m$, let $\bm n_i\in\mathbb R^d$ be its unit direction vector pointing away from $v$, let $f_i$ be the force magnitude, and let $\Lambda_i$ be its stretch ratio. If the directions are unchanged during loading (i.e.\ $\bm n_i$ are constant) and the direction vectors satisfy
\begin{equation}\label{eq:kernel1d}
\dim\ker A = 1,\qquad A:=\big[\bm n_1\ \bm n_2\ \cdots\ \bm n_m\big]\in\mathbb R^{d\times m}.
\end{equation}
then the force ratios are constant during loading:
\begin{equation}\label{eq:force-ratio-const}
\frac{f_i}{f_j}=\text{const},\qquad \forall i,j.
\end{equation}
Moreover, if each spring satisfies the monomial constitutive law
\begin{equation}\label{eq:mono-law-restate}
  f(\Lambda)=(\Lambda-1)^p,\qquad p>0,
\end{equation}
then the strains $\varepsilon_i:=\Lambda_i-1$ also satisfy constant proportionality ratios during loading:
\begin{equation}\label{eq:strain-ratio-const}
  \frac{\varepsilon_i}{\varepsilon_j}=\text{const},\qquad \text{for all }i,j,
\end{equation}
\end{theorem}

\begin{proof}
We first prove the constant force ratios \eqref{eq:force-ratio-const}. With $A$ defined in \eqref{eq:kernel1d}, the force equilibrium condition is the homogeneous linear system
\[
A\,\bm f=\bm 0,\qquad \bm f:=(f_1,\dots,f_m)^\top.
\]
Since the directions $\bm n_i$ are unchanged during loading, the matrix $A$ is constant.
By \eqref{eq:kernel1d}, $\ker A$ is one-dimensional. Hence there exists a nonzero vector $\bm f^\star\in\ker A$ such that any solution of $A \bm f=\bm 0$ must be of the form
\[
\bm f = \alpha\,\bm f^\star
\]
for some scalar $\alpha$. Therefore, for any $i,j$ with
$f_j^\star\neq 0$,
\[
\frac{f_i}{f_j}=\frac{\alpha f_i^\star}{\alpha f_j^\star}=\frac{f_i^\star}{f_j^\star},
\]
which is independent of the loading. This proves \eqref{eq:force-ratio-const}.

Then we prove the constant strain ratios under the monomial law. By \eqref{eq:mono-law-restate}, $f_i = (\Lambda_i-1)^p = \varepsilon_i^p$. For any $i,j$,
\[
\frac{f_i}{f_j}=\left(\frac{\varepsilon_i}{\varepsilon_j}\right)^p.
\]
Since we have proved $f_i/f_j$ is constant during loading, it follows that
\[
\frac{\varepsilon_i}{\varepsilon_j}=\left(\frac{f_i}{f_j}\right)^{1/p}
\]
is also constant during loading, proving \eqref{eq:strain-ratio-const}.

\end{proof}

\begin{theorem}[Crack-tip slope scaling and convergence of $\Gamma_0$]\label{thm:Gamma0-convergence-scaling}
Fix $p>0$ and an integer $n_0\geq 1$ that is independent of the network size. Consider a square lattice with $N=2n$ vertical layers, where $n\geq n_0+1$. Assume that:
\begin{enumerate}
  \item[(i)] For each layer $i=1,2,\dots,n$, the chain stretch satisfies
  \begin{equation}
    \Lambda_i-1=k_i(\lambda-1),
  \end{equation}
  where $k_i>0$ is independent of $\lambda$.
  \item[(ii)] The idealized two-regime slope ratios corresponding to Eqs.~\eqref{eq:slope-ratio-inner-regime} and \eqref{eq:slope-ratio-outer-regime} are
  \begin{equation}\label{eq:two-regime-slope-ratios-appendix}
    \frac{k_i}{k_{i+1}}
    =
    \begin{cases}
      \displaystyle\left(1+\frac{1}{i}\right)^{1/p},
      &1\leq i\leq n_0,\\[6pt]
      \displaystyle\left(1+\frac{1}{i}\right)^{1/(p+1)},
      &n_0<i\leq n-1.
    \end{cases}
  \end{equation}
  \item[(iii)] The slopes satisfy the geometric height constraint
  \begin{equation}\label{eq:height-constraint-appendix}
    \frac{1}{2}k_1+\sum_{i=2}^{n}k_i=n-\frac{1}{2}.
  \end{equation}
\end{enumerate}
Then, as $N\to\infty$ with $n_0$ fixed, the crack-tip slope satisfies
\begin{equation}\label{eq:crack-tip-slope-two-regime-appendix}
  k_1
  \sim
  N^{1/(p+1)}.
\end{equation}
If crack initiation occurs at a fixed crack-tip chain failure stretch ratio $\Lambda_f$, then $\Gamma_0$ converges to a size-independent limit.
\end{theorem}

\begin{proof}
  Iterating the inner-regime ratios in \eqref{eq:two-regime-slope-ratios-appendix} gives
  \begin{equation}\label{eq:inner-slope-recursion-appendix}
    k_i=k_1 i^{-1/p},
    \qquad 1\leq i\leq n_0+1.
  \end{equation}
  For $i\geq n_0+1$, iteration through the inner regime and then through the outer regime gives
  \begin{equation}
    \frac{k_1}{k_i}
    =
    \prod_{j=1}^{n_0}\left(1+\frac{1}{j}\right)^{1/p}
    \prod_{j=n_0+1}^{i-1}\left(1+\frac{1}{j}\right)^{1/(p+1)}
    =
    (n_0+1)^{1/p}\left(\frac{i}{n_0+1}\right)^{1/(p+1)}.
  \end{equation}
  Therefore,
  \begin{equation}\label{eq:outer-slope-recursion-appendix}
    k_i=k_1(n_0+1)^{-1/(p(p+1))}i^{-1/(p+1)},
    \qquad n_0+1\leq i\leq n.
  \end{equation}

  To determine the size dependence of $k_1$, substitute \eqref{eq:inner-slope-recursion-appendix} and \eqref{eq:outer-slope-recursion-appendix} into the height constraint \eqref{eq:height-constraint-appendix}. This yields
  \begin{equation}\label{eq:k1-exact-two-regime-appendix}
    k_1
    =
    \frac{n-\frac{1}{2}}{
      \displaystyle
      \frac{1}{2}
      +\sum_{i=2}^{n_0}i^{-1/p}
      +(n_0+1)^{-1/(p(p+1))}
      \sum_{i=n_0+1}^{n}i^{-1/(p+1)}}.
  \end{equation}
  Because $n_0$ is independent of $n$, the inner-regime contribution remains bounded. The outer-regime sum satisfies
  \begin{equation}
    \sum_{i=n_0+1}^{n}i^{-1/(p+1)}
    =\frac{p+1}{p}n^{p/(p+1)}+O(1).
  \end{equation}
  Substitution into \eqref{eq:k1-exact-two-regime-appendix} proves \eqref{eq:crack-tip-slope-two-regime-appendix}.

  At crack initiation, the linear stretch relation for the crack-tip chain gives
  \begin{equation}
    \Lambda_f-1=k_1(\lambda_c-1),
  \end{equation}
  which implies that
  \begin{equation}
    \lambda_c-1=\frac{\Lambda_f-1}{k_1}.
  \end{equation}
  Hence
  \begin{equation}
    \Gamma_0
    =\frac{N-1}{p+1}(\Lambda_f-1)^{p+1}k_1^{-(p+1)}.
  \end{equation}
  Substitution of \eqref{eq:crack-tip-slope-two-regime-appendix} gives the size-independent limit of $\Gamma_0$ as $N\to\infty$.
\end{proof}

\vspace*{0pt}
\begin{figure}[H]
  \centering
  \includegraphics[trim={0cm 0cm 0cm 0cm},clip, width=1.0\textwidth]{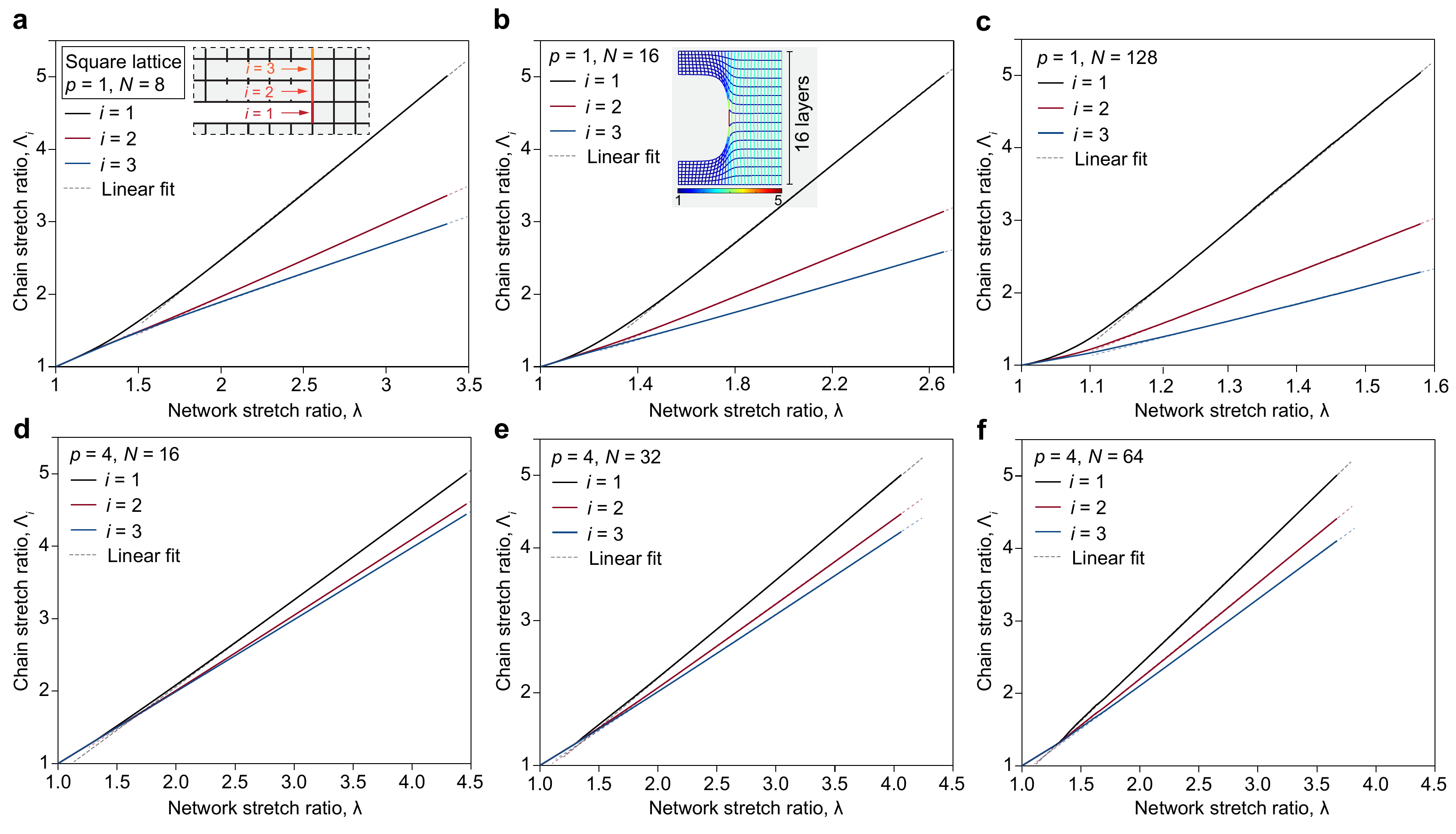}
    \caption[\textbf{Simulation results for chain stretch ratio versus network stretch ratio in a square lattice network.}]
    {\textbf{Simulation results for chain stretch ratio versus network stretch ratio in a square lattice network.} \textbf{(a)} Square lattice network with network layers of $N = 8$ and linear chain $p = 1$. \textbf{(b)} Network with $N = 16$ and linear chain. \textbf{(c)} Network with $N = 128$ and linear chain. \textbf{(d)} Network with $N = 16$ and non-linear chain $p = 4$. \textbf{(e)} Network with $N = 32$ and $p = 4$.\textbf{(f)} Network with $N = 64$ and $p = 4$.
  }
  \label{fig-supp:Figure_SI_1}
\end{figure}

\vspace*{0pt}
\begin{figure}[H]
  \centering
  \includegraphics[trim={0cm 0cm 0cm 0cm},clip, width=0.9\textwidth]{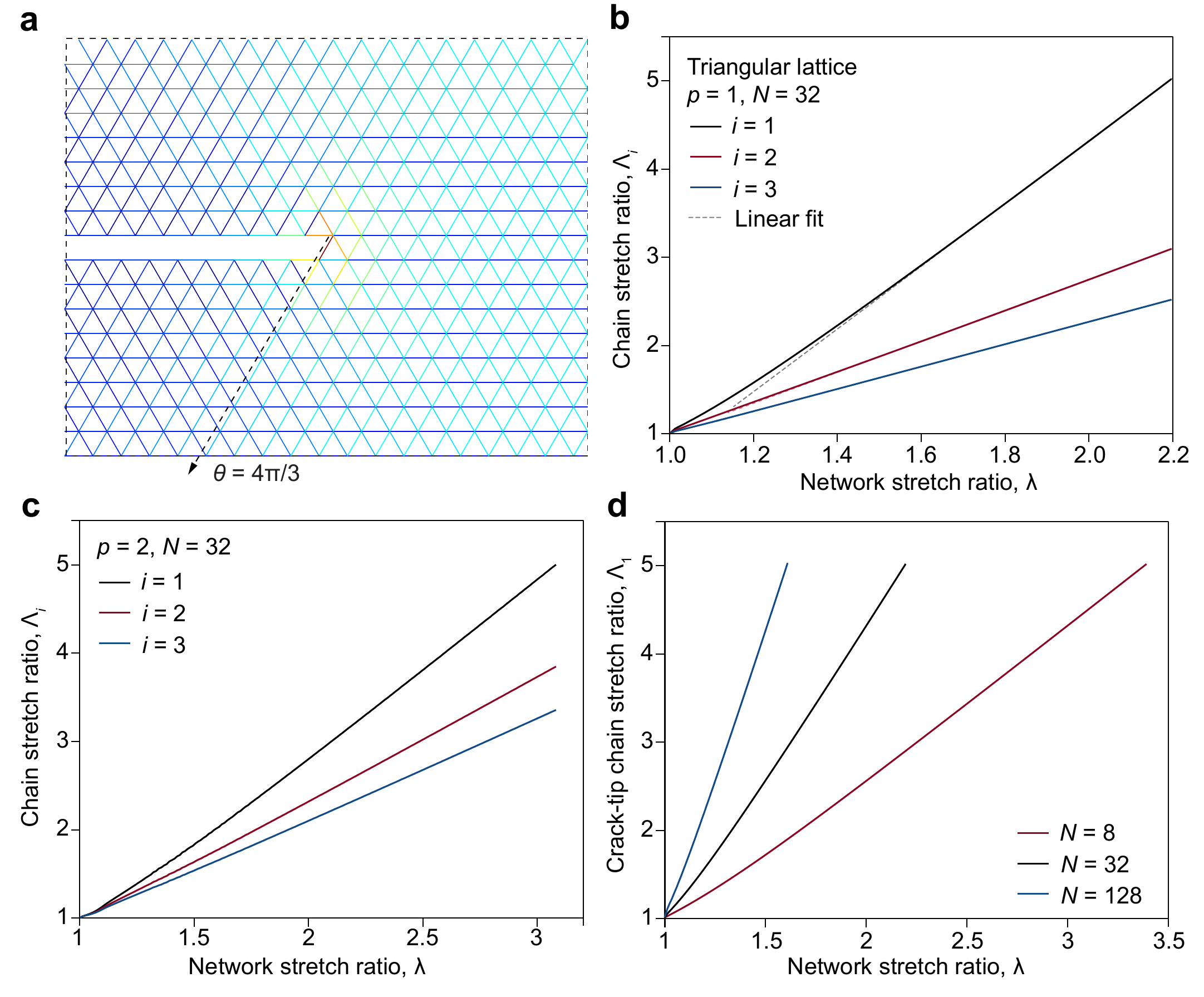}
  \caption[\textbf{Simulation results for chain stretch ratio versus network stretch ratio in a triangular lattice network.}]
  {\textbf{Simulation results for chain stretch ratio versus network stretch ratio in a triangular lattice network.} \textbf{(a)} A deformed triangular lattice network shown in undeformed configuration. \textbf{(b)} Network with network layers of $N = 32$ and linear chain $p = 1$. \textbf{(c)} Network with network layers of $N = 32$ and nonlinear chain $p = 2$. \textbf{(d)} Crack-tip chain stretch ratio versus network stretch ratio with different layers $N = 8$, 32, 128 and $p=1$.
  }
  \label{fig-supp:Figure_SI_2}
\end{figure}

\vspace*{0pt}
\begin{figure}[H]
  \centering
  \includegraphics[trim={0cm 0cm 0cm 0cm},clip, width=0.9\textwidth]{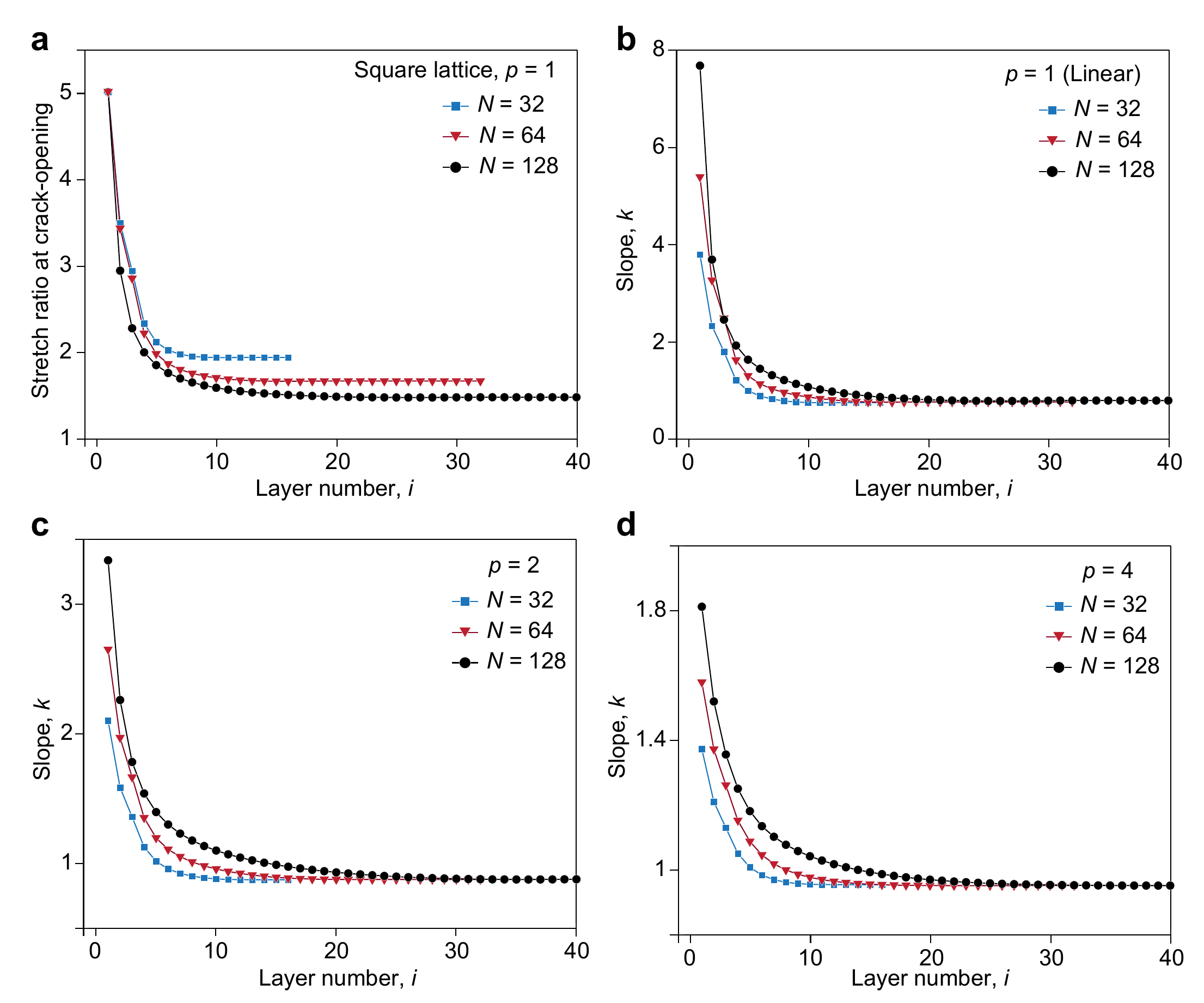}
  \caption[\textbf{Layer-wise slope of square lattice network along discrete HRR line.}]
  {\textbf{Layer-wise slope of square lattice network along discrete HRR line.} \textbf{(a)} Stretch ratio versus layer number when crack initialization. \textbf{(b)} The slopes $k$ versus layer number $i$ of network with linear chain. \textbf{(c)} The slopes $k$ versus layer number $i$ of network with non-linear chain $p = 2$. \textbf{(d)} The slopes $k$ versus layer number $i$ of network with non-linear chain $p = 4$.
  }
  \label{fig-supp:Figure_SI_3}
\end{figure}

\vspace*{0pt}
\begin{figure}[H]
  \centering
  \includegraphics[trim={0cm 0cm 0cm 0cm},clip, width=0.9\textwidth]{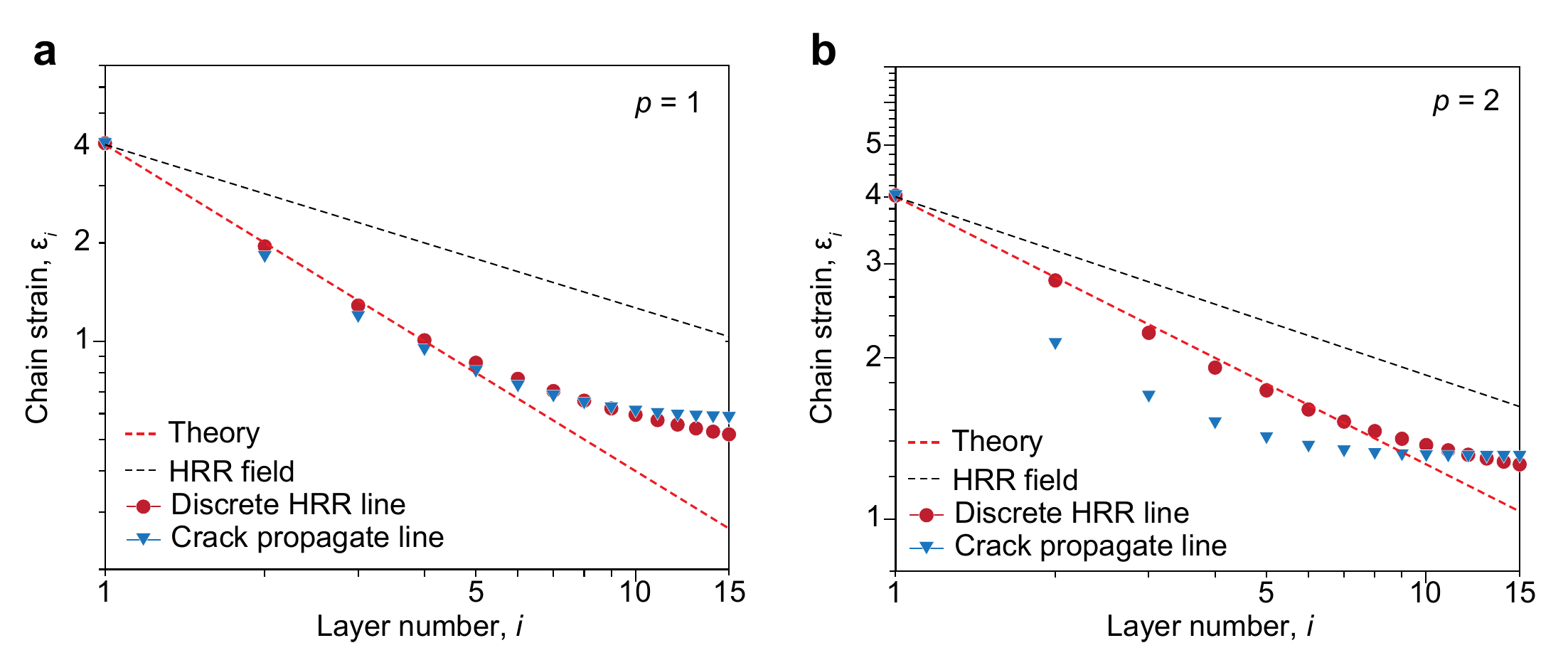}
  \caption[\textbf{Chain-strain decay along the vertical discrete HRR line and horizontal crack-propagation line of square lattice. }]
  {\textbf{Chain-strain decay along the vertical discrete HRR line and horizontal crack-propagation line of square lattice.} \textbf{(a)} Linear constitutive law $p$ = 1. \textbf{(b)} Non-linear constitutive law $p$ = 2.}
  \label{fig-supp:Figure_SI_7}
\end{figure}

\vspace*{0pt}
\begin{figure}[H]
  \centering
  \includegraphics[trim={0cm 0cm 0cm 0cm},clip, width=0.75\textwidth]{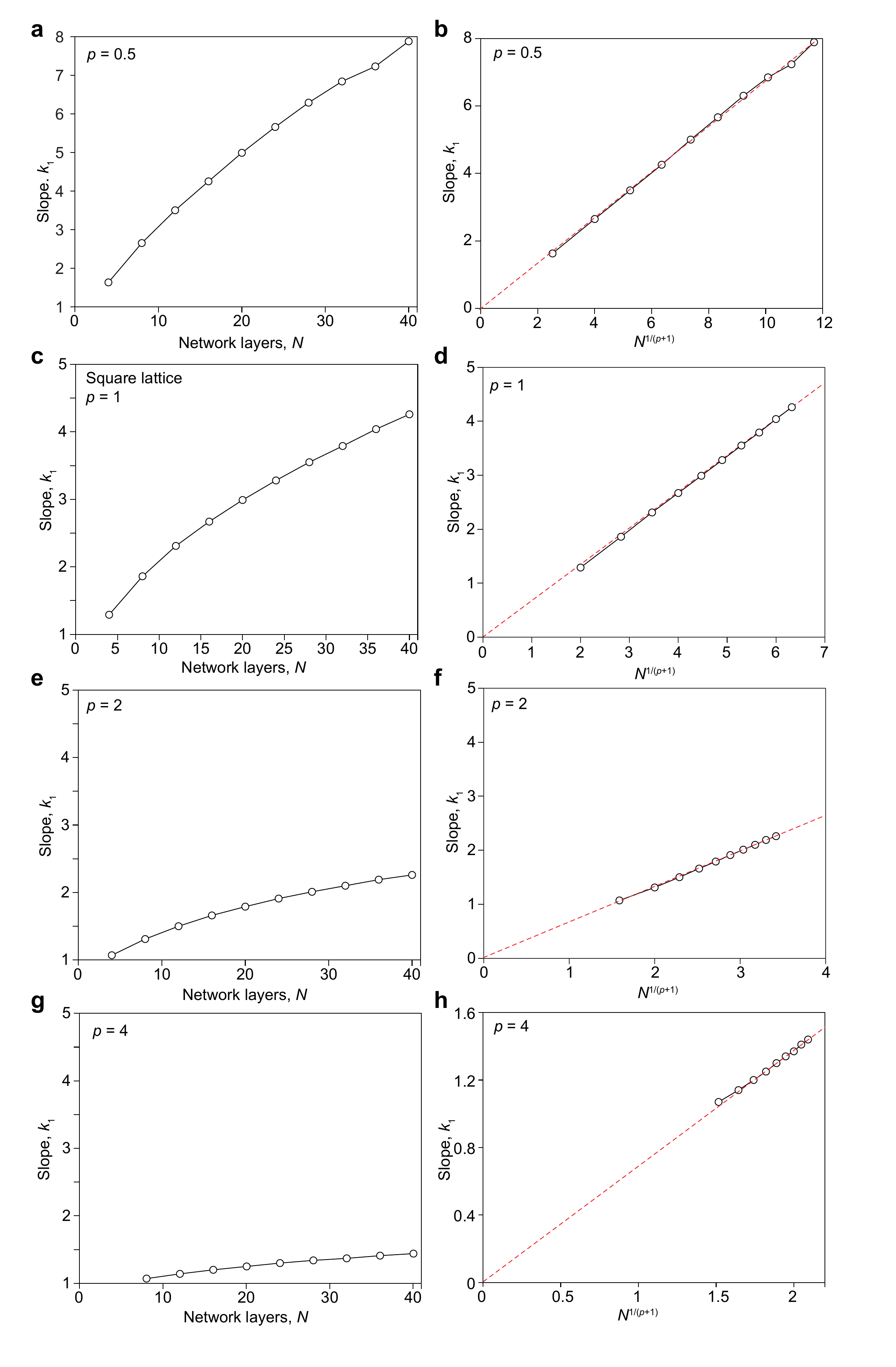}
  \caption[\textbf{The slope of square lattice network versus the network layer $N$.}]
  {\textbf{The slope of square lattice network versus the network layer $N$.} \textbf{(a, c, e, g)} The slope $k$ versus the network layers $N$ with different chains $p = 0.5$, 1, 2, 4. \textbf{(b, d, f, h)} The corresponding slope $k$ versus $N^{1/(p+1)}$ with different chains $p = 0.5$, 1, 2, 4.
  }
  \label{fig-supp:Figure_SI_4}
\end{figure}

\vspace*{0pt}
\begin{figure}[H]
  \centering
  \includegraphics[trim={0cm 0cm 0cm 0cm},clip, width=1.0\textwidth]{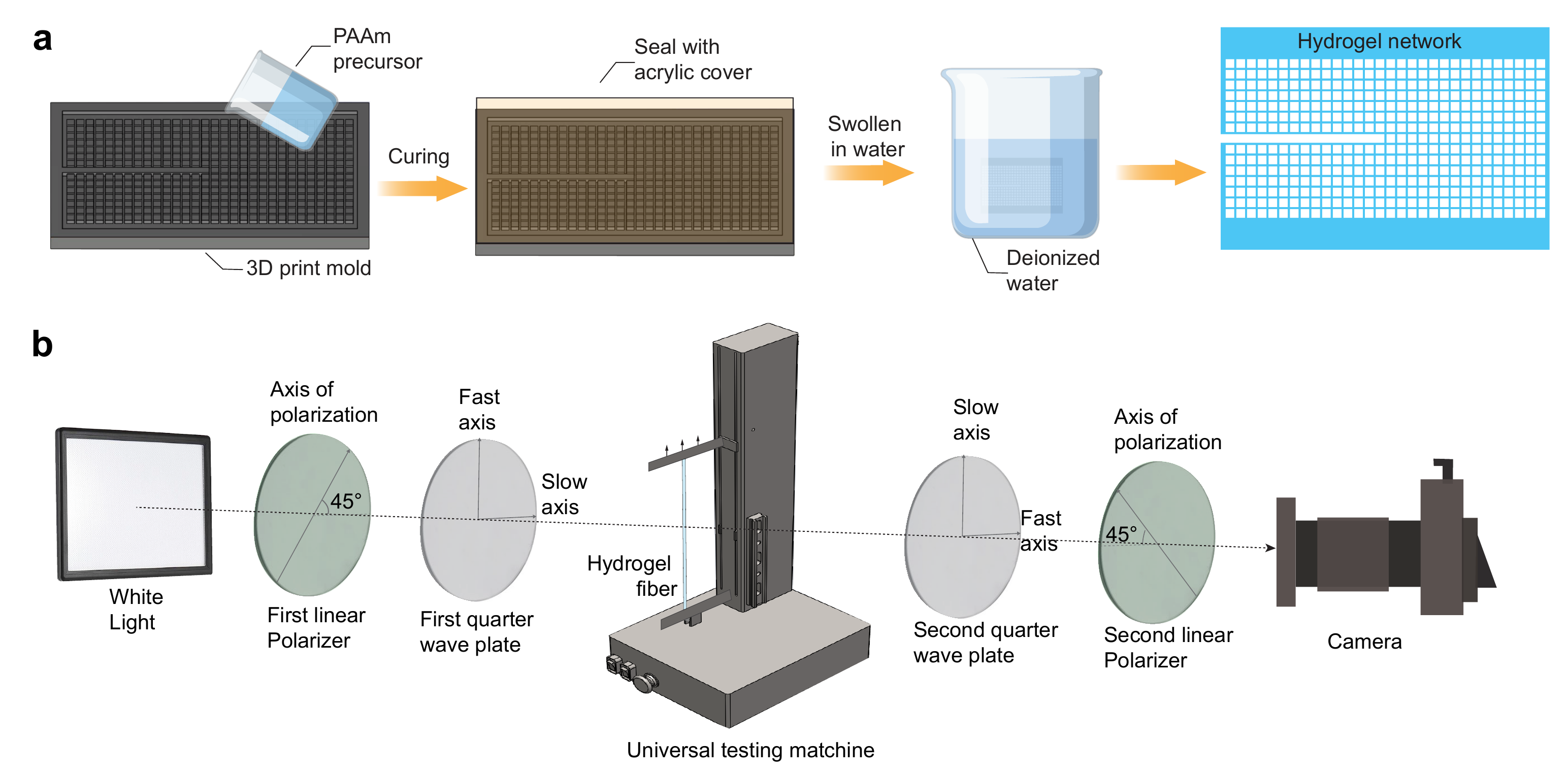}
  \caption[\textbf{Procedure for synthesizing the hydrogel network and schematic of the homemade circular polariscope setup.}]
  {\textbf{Procedure for synthesizing the hydrogel network and schematic of the homemade circular polariscope setup.} \textbf{(a)} Schematic of the procedure for synthesizing the hydrogel network. \textbf{(b)} Schematic of the homemade photoelastic setup, which consists of a white light, a camera, two linear polarizers, two quarter-wave plates, a universal testing machine, and a hydrogel fiber.
  }
  \label{fig-supp:Exp_1}
\end{figure}

\vspace*{0pt}
\begin{figure}[H]
  \centering
  \includegraphics[trim={0cm 0cm 0cm 0cm},clip, width=1.0\textwidth]{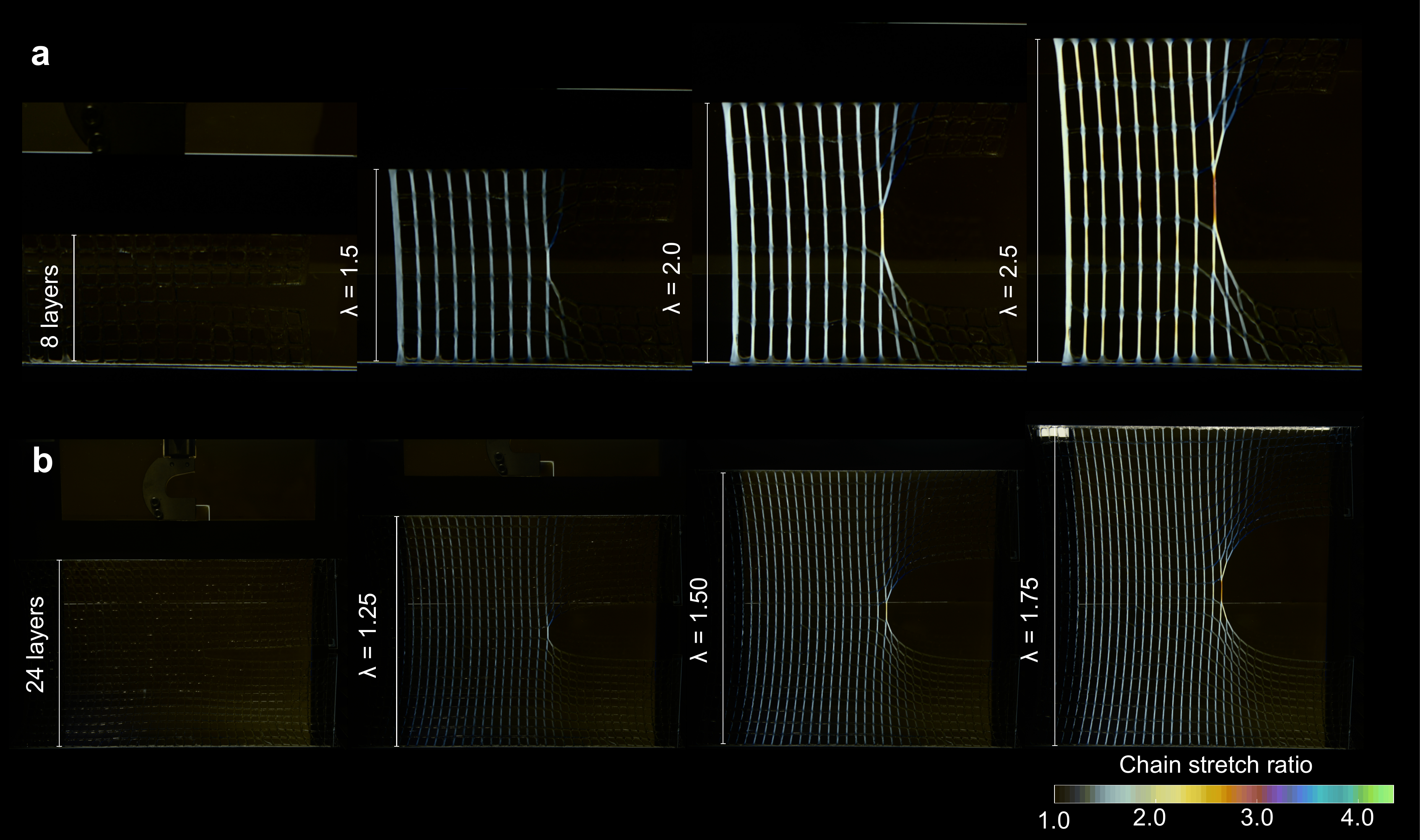}
  \caption[\textbf{Images of hydrogel network with different stretch ratio captured under polarized light.}]
  {\textbf{Images of hydrogel network captured under polarized light.} \textbf{(a)} Hydrogel network with network layer $N = 8$. \textbf{(b)} Hydrogel network with network layer $N = 24$. 
  }
  \label{fig-supp:Exp_Square}
\end{figure}

\end{document}